\documentclass[conference]{IEEEtran}

\usepackage{cite}

\ifCLASSINFOpdf
  
\else
 
\fi
\usepackage{tikz}
\usepackage{subfig}
\usepackage{amsmath, amsthm, amssymb}
\usepackage{graphicx}
\usepackage{multirow}
\usepackage{slashbox}
\usepackage{colortbl}
\usepackage{tabularx} 

\newtheorem{theorem}{Theorem}
\newtheorem{lemma}[theorem]{Lemma}
\newtheorem{proposition}[theorem]{Proposition}

\newtheorem{definition}[theorem]{Definition}

\newtheorem{Open question}{Open Question}
\usepackage{amsfonts}
\usepackage{amsmath}
\begin{document}

\title{Multi-games and a double  game extension of the Prisoner's Dilemma}

\author{\IEEEauthorblockN{Abbas Edalat}
\IEEEauthorblockA{Department of Computing\\
Imperial College London\\
London, UK\\
Email: a.edalat@imperial.ac.uk}
\and
\IEEEauthorblockN{Ali Ghoroghi}
\IEEEauthorblockA{Department of Computing\\
Imperial College London\\
London, UK\\
Email: ali.ghoroghi@imperial.ac.uk}
\and
\IEEEauthorblockN{Georgios Sakellariou}
\IEEEauthorblockA{Department of Computing\\
Imperial College London\\
London, UK\\
Email: georgios.sakellariou06@imperial.ac.uk}}

\maketitle

\begin{abstract}

We propose a new class of games, called Multi-Games (MG), in which a given number of players play a fixed number of basic games simultaneously.  Each player can have different sets of strategies for the different basic games. In each round of the MG, each player will have a specific set of weights, one for each basic game, which add up to one and represent the fraction of the player's investment in each basic game. The total payoff for each player is then the convex combination, with the corresponding weights, of the payoffs it obtains in the basic games.  The basic games in a MG can be regarded as different environments for the players, and, in particular, we submit that MG can be used to model investment in a global economy with different national or continental markets. When the players' weights for the different games in MG are private information or types with given conditional probability distributions, we obtain a particular class of Bayesian games. We show that for the class of so-called completely pure regular Double Game (DG) with finite sets of types, the Nash equilibria (NE) of the basic games can be used to compute a Bayesian Nash equilibrium of the DG in linear time with respect to the number of types of the players.  We study a DG for the Prisoner's Dilemma (PD) by extending the PD  with a second so-called Social Game (SG), generalising the notion of altruistic extension of a game in which players have different altruistic levels (or social coefficients). We show that, with respect to the SG we choose, the payoffs for PD give rise to two different types of DG's. We study two different examples of Bayesian games in this context in which the social coefficients have a finite set of values and each player only knows the probability distribution of the opponent's social coefficient. In the first case we have a completely pure regular DG for which we deduce a Bayesian NE. Finally, we use the second example to compare various strategies in a round-robin tournament of the DG for PD, in which the players can change their social coefficients incrementally from one round to the next.

\end{abstract}
\IEEEpeerreviewmaketitle

\section{Introduction}
Game theory was originally introduced to model the behaviour of rational agents in which players make independent decisions in order to maximise their utility or payoff in an economy~\cite{neumann1944}.  The notion of Nash equilibrium (NE) has become the key concept in game theory since John Nash's celebrated proof of existence of a mixed Nash equilibrium for all finite games in 1950~\cite{Nash1950}. A similar notion of Bayesian NE is also at the basis of games with incomplete information as shown by Harsanyi in 1960's~\cite{Harsanyi1995}.  

However, such an equilibrium strategy for players can only be useful for determining or predicting economic behaviour if it can be efficiently computed, whereas it has become clear in a number of papers that computation of a NE or even an approximate $\epsilon$-Nash equilibrium is in general a computationally hard problem~\cite{Daskalakis06,Etessami10}. Clearly, if the computation of a NE is unfeasible because of its high complexity, then its existence, despite having theoretical significance, has no value in practice. 
It is thus useful to have models in game theory for which the computation of a Nash equilibrium can be more efficiently done than in general. 

In this paper, we propose a class of games called Multi-Games (MG) which can be used to model economic behaviour when each player can allocate its resources in varying proportions to play in a number of different environments, each representing a basic game in its own right.  Each player can have different sets of strategies for the different basic games. The payoff of each player in a MG is assumed to be the convex linear combination of payoffs obtained for the basic games weighted by the allocated proportions to them. Here we first present a simple example of a Double Game (DG) with two basic games. 

Consider two multinational companies which can invest in the national economies of two different countries $R_1$ and $R_2$ with different cost of investment, rates of profit, labour value, interest rates etc. Suppose they need to decide in what ratio to divide their assets for investment in the two countries and, in addition, whether to enter into a particular venture or simply deposit their allocated assets in some bank in each country. We thus have two games, $G_1$ for $R_1$ and $G_2$ for $R_2$, one for each national economy, each with two players and two strategies for entering ($E$) and not entering ($N$).  Let $\pi_{ij}$ denote the payoff function in $G_i$ for player $j$ (with $i,j=1,2$). 

Suppose the first player invests $\lambda$ and $1-\lambda$ fractions of its assets in $R_1$ and $R_2$ respectively, and assume $\gamma$ and $1-\gamma$ are the corresponding fractions for the second player. Then, the payoff to the first player for the strategy profiles $(X_1,Y_1)$ in $G_1$ and $ (X_2,Y_2)$ in $G_2$, with $X_i,Y_i\in \{E,N\}$ for $i=1,2$,  would be \[\lambda \pi_{11}(X_1,Y_1) +(1-\lambda)\pi_{21}(X_2,Y_2),\] whereas the payoff for the second player for the same strategy profile would be \[\gamma \pi_{12}(X_1,Y_1) +(1-\gamma)\pi_{22}(X_2,Y_2).\] Now if the coefficients $\lambda$ and $\gamma$ are both private information and can each take only a finite number of values between zero and one, then the DG is reduced to a Bayesian game with a finite set of types for the two players and we can look for a Bayesian Nash equilibrium. The values $0$ and $1$ will always be included in the set of possible types of each player and are called the {\em extreme} types. 

In the next section, we will define MG in general and then, for convenience and ease of presentation, we restrict ourselves to the class of 2-player DG in which each player has the same  set of strategies in the two basic games. We later define the class of {\em pure regular} DG in which for the pairs of  extreme types there are four pure NE in which the strategy of each player only depends on its own type.  Similarly, for a DG with $k$ and $\ell$ types for the two players respectively, we define the notion of a  {\em completely pure regular} DG where there are $k\times \ell$ pure NE for all possible pairs of types of the two players in which the strategy of each player only depends on its own type. We then derive a test for establishing that a DG  is completely pure regular in linear time with respect to the number of types of the players and show that a pure Bayesian equilibrium for completely pure regular DG can be obtained directly from this test, thus reducing the complexity of computation. 

We will then apply this framework to obtain a double game extension of the Prisoner's Dilemma (PD) to model prosocial behaviour. In this DG for PD, the first game is the classical PD and the second game captures the social or moral gain for cooperation for each player. We now review the concept of prosocial behaviour and moral gain in PD. 

The PD is considered a standard method for modelling social dilemmas~\cite{OB2007,SG1970}. In the 1980's, Axelrod organised two international round-robin tournaments in which strategies for the repeated PD competed with each other~\cite{AE1980,AM1980}. In the competition, tit-for-tat, i.e., cooperate on the first move and then reciprocate the opponent's last move, proved to be robust and became the overall winner of the tournaments~\cite{AE1980,AM1980}. Axelrod then promoted tit-for-tat, and the four associated characteristics of (i) be nice, (ii) reciprocate, (iii) don't be envious, (iv) don't be too clever, as the way reciprocal altruism has evolved~\cite{The84}. The PD has also been used to model conditional altruistic cooperation, which has also been tested by real monetary payoffs~\cite{Gintis2009}. 

However, when confronted with the choice to cooperate or defect, human beings not only consider their material score, but also the social and moral payoffs of any decision they make. This means that the material payoffs presented in the PD cannot provide a complete picture of the decision making process human beings follow. In fact, according to some researchers, human social evolution has a moral direction which has extended our moral campus in the course of thousands of years of increasing communication and complexity in human social organisation~\cite{wright2000,wright2009}. Moreover, there are individual and temporal variations in pro-social attitudes of human beings with some making decisions more based on self-interest than others. A more adequate model of human behaviour should take into account these aspects of social evolution as well. The same applies to economic decisions by corporations or governments, in which actions taken can have significant social and environmental implications, which are not incorporated in the material gains they produce. In~\cite{sheng94}, it was proposed that a {\em coefficient of morality} be introduced to the PD and the payoff values of the players be accordingly changed. The so-called {\em altruistic extension} of any finite strategic game was defined in~\cite{Chen2011}, which endows each player with an altruistic level in the unit interval which provides the weight of the pro-social attitude of the player.  This modification aims to reflect real-life situations and dilemmas more accurately by taking into account both material and moral/social gains. Thus, for each player, the payoff is a weighted, linear combination of the payoffs for the PD and the SG. In essence, it is a convex combination of the payoffs occurring from both material and social dilemmas.

In this paper, we show that the DG, as an instance of MG, provides a generalisation of the altruistic extension in~\cite{Chen2011} which can be considered as a DG with the first game identified as the original game and the second game as a symmetric altruistic game. In a general DG, the social or altruistic game is allowed to be non-symmetric, which means that in general the altruistic payoffs for the different players may be different even for the same strategy profile. 

We furthermore consider the DG for the PD where the social (altruistic) coefficient of each player forms a finite discrete set of incomplete information or types thus giving rise to a Bayesian game. We prove that this DG is in fact pure regular and determine its Bayesian equilibrium when it is completely pure regular.

\section{Multi Games}

The purpose of using this model is to add a new dimension to the description of a range of situations, achieved through the employment of game theoretic models. This is done by linearly combining the payoff matrices of various games and linking them through the use of a coefficient for each player, which represents the amount of investment that a player is willing to commit in that particular game. More specifically MG is defined as follows. Consider $M$ finite $N$-player games $G_i$ ($1\leq i\leq M$) with the strategy set $S_{ij}$ and payoff matrix $\pi_{ij}$ for player $j$ ($1\leq j\leq N$) in the game $G_j$. Assume each player $j$ is equipped with a set of $M$ weights $\lambda_{ij}$ with $\sum_{j=1}^M\lambda_{ij}=1$. We define the MG N-player game $G$ with basic games $G_i$ as the finite strategy game with players $j$  ($1\leq j\leq N$) each having strategy set $\prod_{1\leq i\leq M}S_{ij}$ and payoff \[\pi_j(\prod_{1\leq i\leq M}s_{ij})=\sum_{1\leq i\leq M}\lambda_{ij}\pi_{ij}(s_{ij})\] for strategy profile $s_{ij}\in S_{ij}$ and possibly incomplete information (types) $\lambda_{ij}$ for $1\leq i\leq M$. We say the MG is {\em uniform} if for each player $j$, the set $S_{ij}$ is independent of the game $G_i$, i.e., we can write $S_{ij}=S_j$, and $s_{ij}=s_{i'j}$ for $1\leq i\leq i'\leq M$.

Harsanyi~\cite{Harsanyi1967-8} suggested a method to solve games with incomplete information by transforming it to a game with imperfect information, in which a probability distribution for each unknown value, referred to as a {\em type}, is provided. Such a game is called a Bayesian game and a Bayesian Nash equilibrium is then a Nash equilibrium in a Bayesian game~\cite{zz}.

In this extended abstract we restrict ourselves to uniform MG with two players $N=2$. We now give a simple example in this case . Assume that the strategy set for each player consists of actions C and D and denote the weights for players 1 and 2 respectively by $\lambda_{ij}$ and $\gamma_{ij}$, with $1\leq i\leq M$ and $j=1, 2$. If the payoff matrix for the basic game $G_i$ is given as in TABLE~\ref{tab:1}, then the payoff matrix for the MG $G$ will be given as in TABLE~\ref{tab:2}.

\begin{table}[ht]\renewcommand{\arraystretch}{2}
\centering
\begin{tabular}{cccc}
\multicolumn{1}{c}{} & \multicolumn{1}{c}{}& \multicolumn{2}{c}{Player 2} \\ 
\cline{3-4}
\multicolumn{1}{c}{} & &\multicolumn{1}{|c|}{C}  &\multicolumn{1}{c|}{D}  \\ 
\cline{2-4}
\multirow{2}{*}{Player 1} & \multicolumn{1}{|c}{C} & \multicolumn{1}{|c|}{($a_{i1}, a_{i2}$)} & \multicolumn{1}{c|}{($b_{i1}, b_{i2}$)} \\ 
\cline{2-4}
& \multicolumn{1}{|c}{D} & \multicolumn{1}{|c|}{($c_{i1}, c_{i2}$)} & \multicolumn{1}{c|}{($d_{i1}, d_{i2}$)}\\
\cline{2-4}
\end{tabular}\footnotesize\caption{Payoff matrix of the basic games $G_i$}\label{tab:1}
\end{table}

\begin{table}[ht]\renewcommand{\arraystretch}{1.7}
\centering
\begin{tabular}{cccc}
\multicolumn{1}{c}{} & \multicolumn{1}{c}{}& \multicolumn{2}{c}{} \\
\cline{3-4}
\multicolumn{1}{c}{} & & \multicolumn{1}{|c|}{C} & \multicolumn{1}{c|}{D} \\ 
\cline{2-4}
\multirow{2}{*}{} & \multicolumn{1}{|c|}{C} &\multicolumn{1}{|@{}c@{}|}{$\sum_{i=1}^M\lambda_i$$a_{i1}$ , $\sum_{i=1}^M\gamma_i$$a_{i2}$}  & \multicolumn{1}{|@{}c@{}|}{ $\sum_{i=1}^M\lambda_i$$b_{i1}$,  $\sum_{i=1}^M\gamma_i$$b_{i2}$ } \\ 
\cline{2-4}
& \multicolumn{1}{|c|}{D} & \multicolumn{1}{|@{}c@{}|}{ $\sum_{i=1}^M\lambda_i$$c_{i1}$,  $\sum_{i=1}^M\gamma_i$$c_{i2}$} & \multicolumn{1}{@{}c@{}|}{ $\sum_{i=1}^M\lambda_i$$d_{i1}$,  $\sum_{i=1}^M \gamma_i$$d_{i2}$} \\
\cline{2-4}
\end{tabular}\scriptsize\caption{Payoff matrix of the MG}\label{tab:2}
\end{table}

\subsection{Coherent pairs of NE}
For finite sets of types, we assume that the weights $\lambda_{ij}$ for each player $j$ are
selected from a finite discrete set, which would denote the types of
each player. When these values are private information we have a
Bayesian game.

For simplicity in this paper, we suppose our MG is
uniform with $M=2$ and $N=2$; our results however do extend to the
more general case. Since we now have two basic games $G_1$ and $G_2$ and two players,
we can denote the weights of the first player by $1-\lambda$ and
$\lambda$, and those of the second player by $1-\gamma$ and
$\gamma$. This means that for $\lambda=0$ the first player invests totally in the first game $G_1$ whereas when  $\lambda=1$, the first player invests totally in the second game $G_2$. Similarly for the second player with weight $\gamma$. 

Thus, in the finite discrete case, the finite set of types of each player is 
given by a set of increasing values, say $\lambda_m$ ($1\leq m\leq k$)
for $\lambda$ and a set of increasing values, say $\gamma_n$ ($1\leq
n\leq \ell$) for $\gamma$, where each type is restricted to the unit
interval. We assume in the discrete case that we always have $0$ and $1$ as types for each
player, i.e., $\lambda_1=\gamma_1=0$ and $\lambda_k=\gamma_\ell=1$, which we call the {\em extreme} types. We let $G^{(\lambda,\gamma)}$ denote the DG game $G$ with the types taking the specific values $\lambda$ and $\gamma$. In the discrete case, we in addition let $G^{mn}$ denote the DG game $G$ with the types $\lambda_m$ and $\gamma_n$ selected for the two players respectively. We refer to a NE for $G^{(\lambda,\gamma)}$ as a {\em local} NE for the DG $G$.

Assume we have a uniform two-player DG $G$ with basic games $G_1$ and $G_2$. Given a player $j$, we denote as usual the strategy set of the opponent of $j$ by $S_{-j}$. 

\begin{definition}
The DG $G$ has a {\em coherent} pair of {\em pure NE} for a player $j$ with a given type if there is an action $s\in S_j$ and actions $u,v\in S_{-j}$ such that $(s,u)$ and $(s,v)$ are respectively pure NE for 
$G$ with the given type for player $j$ and the two extreme types of the other player. In this case, the pair of profiles $((s,u),(s,v))$ is called the {\em coherent pair} of pure NE for player $j$ with the given type.
\end{definition}
For example, suppose we fix a type $\lambda_m$ with $1\leq m\leq k$ for the first player, then the DG has a coherent pair of pure strategies for the first player with type $\lambda_m$ if 
there are actions $s\in S_1$ and $u,v\in S_2$ such that $(s,u)$ is a pure NE for $G^{m1}$ and $(s,v)$ is a pure NE for $G^{m\ell}$.

\begin{proposition}\label{pure-coherent}
If the DG has a coherent pair  $((s,u),(s,v))$ of pure NE for the first player with type $\lambda_m$, then there exists an integer $p$ with $1\leq p\leq \ell$ such that  $G^{mn}$ has $(s,u)$ as a pure NE for $1\leq n\leq p$ and has $(s,v)$ as a pure NE for $p<n\leq \ell$.
\end{proposition}

There is also a version of the above result for mixed NE, which holds for the general case of continuous types $\lambda$ for player 1 and $\gamma$ for player 2. We denote the set of mixed strategies over the strategy set $S$ by ${\cal M}(S)$.

\begin{definition}
The DG $G$ has a {\em coherent} pair of  {\em mixed NE} for a player $j$ with a given type if there is a mixed strategy $\sigma\in {\cal M}(S_i)$ and mixed strategies $\sigma_0,\sigma_1\in {\cal M}(S_{-j}$) such that $(\sigma,\sigma_0)$ and $(\sigma,\sigma_1)$ are respectively mixed NE for
$G$ with the given type for player $j$ and the two extreme types of the other player. In this case, the pair of profiles $((\sigma,\sigma_0),(\sigma,\sigma_1))$ is called the {\em coherent pair} of mixed NE for player $j$ with the given type.
\end{definition}
Thus, player one will have a coherent pair of mixed NE if there exists $\sigma\in {\cal M}(S_1)$ and  $\sigma_0,\sigma_1\in {\cal M}(S_2)$ such that $(\sigma,\sigma_0)$ is mixed NE for $G^{\lambda,0}$ and  $(\sigma,\sigma_1)$ is mixed NE for $G^{\lambda,1}$. 

\begin{proposition}
If the DG with continuous types $\lambda$ and $\gamma$ has a coherent pair of mixed NE $(\sigma,\sigma_0)$ and  $(\sigma,\sigma_1)$ for the first player with a given type $\lambda$, then for each $\gamma$ in the unit interval there exists a mixed strategy profile of the form $(\sigma,\sigma_\gamma)$ that is a NE for $G^{(\lambda,\gamma)}$.
\end{proposition}

\begin{proof}
For simplicity we present the proof for the case that each player has a strategy set with two actions: $S_1=\{C_1,D_1\}$ and $S_2=\{C_2,D_2\}$, with the payoff matrices $G^{(\lambda,0)}$ and $G^{(\lambda,1)}$ given in TABLE~\ref{tab:3} and TABLE~\ref{tab:4}.

\begin{table}[ht]\renewcommand{\arraystretch}{2}
\centering
\begin{tabular}{cccc}
\multicolumn{1}{c}{} & \multicolumn{1}{c}{}& \multicolumn{2}{c}{Player 2} \\ 
\cline{3-4}
\multicolumn{1}{c}{} & &\multicolumn{1}{|c|}{C}  &\multicolumn{1}{|c|}{D}  \\ 
\cline{2-4}
\multirow{2}{*}{Player 1} & \multicolumn{1}{|c|}{C} & \multicolumn{1}{|c|}{($a_1, a_2$)} & \multicolumn{1}{|c|}{($b_1, b_2$)} \\ 
\cline{2-4}
& \multicolumn{1}{|c|}{D} & \multicolumn{1}{|c|}{($c_1, c_2$)} & \multicolumn{1}{|c|}{($d_1, d_2$)}\\
\cline{2-4}
\end{tabular}\footnotesize\caption{PD - Payoff Matrix Representation}\label{tab:3}
\end{table}

\begin{table}[ht]\renewcommand{\arraystretch}{2}
\centering
\begin{tabular}{cccc}
\multicolumn{1}{c}{} & \multicolumn{1}{c}{}& \multicolumn{2}{c}{Player 2} \\ 
\cline{3-4}
\multicolumn{1}{c}{} & &\multicolumn{1}{|c|}{C}  &\multicolumn{1}{|c|}{D}  \\ 
\cline{2-4}
\multirow{2}{*}{Player 1} & \multicolumn{1}{|c|}{C} & \multicolumn{1}{|c|}{($e_1, e_2$)} & \multicolumn{1}{c|}{($f_1, f_2$)} \\ 
\cline{2-4}
& \multicolumn{1}{|c|}{D} & \multicolumn{1}{|c|}{($g_1, g_2$)} & \multicolumn{1}{c|}{($h_1, h_2$)}\\
\cline{2-4}
\end{tabular}\footnotesize\caption{PD - Payoff Matrix Representation}\label{tab:4}
\end{table}

Suppose $\sigma=pC_1+(1-p)D_1$ and $\sigma_\gamma=p_\gamma C_2,+(1-p_\gamma)D_2$, where for $\gamma=0$ and $\gamma=1$ we obtain the pair of coherent mixed NE  $(\sigma,\sigma_0)$ and  $(\sigma,\sigma_1)$. Then a long calculation shows that $(\sigma,\sigma_\gamma)$ is a mixed NE for $G^{(\lambda,\gamma)}$ if 

\[{\textstyle p_\gamma=}\]\[{\textstyle \frac{(1-\gamma)p_0[p(a-2-b_2)+(1-p)(c_2-d_2)]+\gamma p_1[p(e_2-f_2)+(1-p)(g_2-h_2)]}{(1-\gamma)[p(a-2-b_2)+(1-p)(c_2-d_2)]+\gamma [p(e_2-f_2)+(1-p)(g_2-h_2)]}.}\]
\normalsize

 The general case can be proved in a similar way.

\end{proof}
\subsection{Pure regular DG}
Next, we examine how information about the set of local pure NE for the DG for various types of the two players can be used to deduce the Bayesian NE for the DG. We say a DG is pure regular if it has a set of four pairs of pure NE for all extreme types, for which the strategy of each player only depends on its own type. Here is the exact definition.
\begin{definition}
We say a DG for two players is {\em pure regular} if there are four  pure strategy profiles $(s,u), (s,v), (t,u), (t,v)$ such that the two pairs $((s,u),(s,v))$ and $((t,u), (t,v))$ are respectively coherent pairs of pure NE for the first player with extreme types $\lambda=0$ and $\lambda=1$ respectively, while the two pairs $((s,u),(t,u))$ and $((s,v),(t,v))$ are  respectively coherent pairs of pure NE for the second player with extreme types $\gamma=0$ and $\gamma=1$ respectively. We say that the four strategy profiles {\em induce} pure regularity. 

\end{definition}
For a DG with a finite set of types for each player, we can go further as follows.

\begin{definition}
We say a DG with finite sets of types given by $\lambda_m$ ($1\leq m\leq k$) and $\gamma_n$ ($1\leq n\leq \ell$) is {\em completely pure regular} if  there are pure strategies $s_m\in S_1$ ($1\leq m\leq k$) and $u_n\in S_2$ ($1\leq n\leq \ell$) such that the strategy profile $(s_m,u_n)$ is a pure Nash equilibrium for the game $G^{mn}$ for $1\leq n\leq \ell$ and $1\leq m\leq k$. 

\end{definition}
It is clear that a completely pure regular DG is pure regular and thus our terminology is consistent. Note also that for a completely pure regular DG as above $(\prod_{1\leq m\leq k}s_m,\prod_{1\leq n\leq \ell}u_n)$ is a pure Bayesian strategy in which the first player takes action $s_m$ for type $\lambda_m$ and the second player takes action $u_n$ for type $\gamma_n$. The above notions can also be extended to {\em mixed regular} and {\em completely mixed regular} DG.  
\begin{lemma}
 The DG is completely pure regular if and only if for all conditional probability distributions for the types of the two players the Bayesian pure strategy $(\prod_{1\leq m\leq k}s_m,\prod_{1\leq n\leq \ell}u_n)$ is a pure Bayesian NE.

\end{lemma}

We can determine if a DG with finite types for the two players is completely pure regular as follows: 
\begin{itemize}
\item [(i)] Test if the DG is pure regular for a set of four pure NE with extreme types.
\item [(ii)] For each set of four pure NE, which induces pure regularity, use Proposition~\ref{pure-coherent} to determine all the pure NE for all $k\times \ell$ combinations of pairs of types for the two players; these will be the pure NE of $G^{mn}$ for $1\leq m\leq k$ and $1\leq n\leq \ell$.
\item [(iii)] Finally check whether the set of  $k\times \ell$ pure NE induces a completely pure regular DG. 
\end{itemize}
These three tasks (i)  - (iii) can be done in linear time (i.e., linear in the maximum number of types $\max(k,\ell)$ for the two players).
We therefore have the following theorem.
\begin{theorem}
Given any DG with finite number of types for the two players, we can decide in linear time if it is completely pure regular, in which case a Bayesian pure NE is obtained in linear time. 

\end{theorem}
The above theorem can be extended to any $N$-player MG with $M$ basic games and finite number of types for each player. A generalisation to completely mixed regular DG is also possible. 
In the next section, we present a pure regular DG with a finite number of types which will give us both an example of a completely pure regular and an example where this property fails.

\section{A DG for Prisoner's Dilemma}

\noindent The PD is a fundamental non-zero-sum problem of game theory that attempts to mathematically analyse the behaviour of individuals in a strategic situation, in which the success of each individual does not depend entirely on one's choice, but on the opponent's as well~\cite{The84}. Essentially, it is an abstract formulation of some common situations in which what is best for each person individually leads to mutual defection, whereas everyone would have been better off with mutual cooperation~\cite{The84}. It has provided a tool for experimental studies in various disciplines, such as economics, social psychology, evolutionary biology and fields that are involved with modelling of social processes, such as behaviour in decision making~\cite{The84}\cite{AHT1981}.

Each of the players competing in the PD has the choice to cooperate (C) or defect (D) and the payoff values gained by the combination of the aforementioned actions are $T, R, P$ and $S$.

The payoff matrix of the PD is presented in TABLE~\ref{tab:5}.

\begin{table}[ht]\renewcommand{\arraystretch}{2}
\centering
\begin{tabular}{cccc}
\multicolumn{1}{c}{} & \multicolumn{1}{c}{}& \multicolumn{2}{c}{Player 2} \\ 
\cline{3-4}
\multicolumn{1}{c}{} & &\multicolumn{1}{|c|}{C}  &\multicolumn{1}{c|}{D}  \\ 
\cline{2-4}
\multirow{2}{*}{Player 1} & \multicolumn{1}{|c|}{C} & \multicolumn{1}{|c|}{($R, R$)} & \multicolumn{1}{c|}{($S, T$)} \\ 
\cline{2-4}
& \multicolumn{1}{|c|}{D} & \multicolumn{1}{|c|}{($T, S$)} & \multicolumn{1}{c|}{($P, P$)}\\
\cline{2-4}
\end{tabular}\footnotesize\caption{PD - Payoff Matrix Representation}\label{tab:5}
\end{table}

The values of $T, R, P$ and $S$ satisfy the following two inequalities:
\begin{center}
$T > R > P > S \;\;$ and $\;\; R > (T + S)/2$
\end{center}

The first equation specifies the order of the payoffs and defines the dilemma, since the best a player can do is get $T$ (i.e. the temptation to defect payoff when the other player cooperates), the worst a player can do is get $S$ (i.e. the sucker's payoff for cooperating while the other player defects), and, in ordering the other two outcomes, $R$ (i.e. the reward payoff for mutual cooperation), is assumed to be better than $P$ (i.e. the punishment payoff for mutual defection)~\cite{The84}. The second equation ensures that, in the repeated game, the players cannot get out of the dilemma by taking turns in exploiting each other. This means that an even chance of exploiting and being exploited is not as good an outcome for a player as mutual cooperation. Therefore, it is assumed that $R$ is greater than the average of $T$ and $S$~\cite{The84}. Finally, a special case of the PD occurs when the apparent advantage of defecting over cooperating is not dependent on the opponent's choice and the disadvantage of the opponent defecting over cooperating is not dependent on one's choice, as can be illustrated in the following equation~\cite{The84}:
\begin{center}
$T + S = P + R \;\; \mbox{or} \;\; R + P - T - S = 0$
\end{center}

\subsection{The Social Game}

\noindent The SG encourages cooperation and discourages defection, as cooperating is usually considered to be the ethical and moral choice to make when interacting with others in social dilemmas. This can be done in different ways corresponding to different types of payoff matrices. Here, we will restrict to the case that the SG encourages cooperation and discourages defection for each player, independently of the action chosen by the other player. 

We present the normal form and the mathematical formulation of the SG as follows. Assume that the competing participants in the SG are player 1 and player 2. Each of them has the choice to select between C and D. When they have both made their choice, the payoffs assigned to them are calculated according to payoff matrix~\ref{tab:6}, where $M_1$, $M_2$ and $M^{'}_{1}$, $M^{'}_{2}$ satisfy:

\[M_1>M'_1,\quad M_2 > M'_{2}\]

When $M_1=M_2$ and $M'_1=M'_2$, we will have a symmetric social game and our framework reduces to the altruistic extension in~\cite{Chen2011}.


\begin{table}[ht]\renewcommand{\arraystretch}{2}
\centering
\begin{tabular}{cccc}
\multicolumn{1}{c}{} & \multicolumn{1}{c}{}& \multicolumn{2}{c}{Player 2} \\ 
\cline{3-4}
\multicolumn{1}{c}{} & &\multicolumn{1}{|c|}{C}  &\multicolumn{1}{|c|}{D}  \\ 
\cline{2-4}
\multirow{2}{*}{Player 1} & \multicolumn{1}{|c}{C} & \multicolumn{1}{|c|}{($M_1, M_2$)} & \multicolumn{1}{|c|}{($M_1, M^{'}_{2}$)} \\ 
\cline{2-4}
& \multicolumn{1}{|c}{D} & \multicolumn{1}{|c|}{($M^{'}_{1}, M_2$)} & \multicolumn{1}{|c|}{($M^{'}_{1}, M^{'}_{2}$)} \\
\cline{2-4}
\end{tabular}\footnotesize\caption{SG - Payoff Matrix Representation}\label{tab:6}
\end{table}

Thus, in the SG we treat in this paper, the players are individually and independently rewarded for cooperating and punished for defecting. This can be interpreted in the following way. Cooperation by an individual, independent of the action of the opponent, is socially rewarded by inducing a good conscience, whereas defection is punished by creating a guilty one. The values of $M_1$, $M_2$  and $M^{'}_{1}$, $M^{'}_{2}$ are assumed to be socially determined to correspond to the average moral norm in the given society and are considered to have evolved in the course of increasing complexity, communication and moral growth in human history.

\subsection{The Double Game}

\noindent Although the payoffs for the SG are determined by the social context of the game, there is still individual variation in pro-social behaviour of the players.  We assume each player has a social coefficient taking values between 0 and 1, which reflects how pro-social they are in practice in each round of the game. In our particular SG, the social coefficient of a player signifies how much the player cares about the morality or the social aspect of their action. The payoffs of the DG for each player are then the weighted sum or convex combination of the payoffs of the PD and SG using the player's social coefficient as represented in TABLE~\ref{tab:7}, where  $\lambda$ and  $\gamma$ (with  $0\leq \lambda, \gamma\leq 1)$ are the social coefficients of players 1 and 2, respectively. Note that  the two players can still play the standard version of the PD by selecting their social coefficients to be equal to 0, in which case the DG reduces to the PD.

\begin{table}[ht]\renewcommand{\arraystretch}{1.7}
\centering
\begin{tabular}{cccc}
\multicolumn{1}{c}{} & \multicolumn{1}{c}{}& \multicolumn{2}{c}{} \\
\cline{3-4}
\multicolumn{1}{c}{} & & \multicolumn{1}{|c|}{C} & \multicolumn{1}{|c|}{D} \\ 
\cline{2-4}
\multirow{2}{*}{} & \multicolumn{1}{|c}{C} &\multicolumn{1}{|@{}c@{}|}{(1 - $\lambda$)$R$ + $\lambda M_1$, (1 - $\gamma$)$R$ + $\gamma M_2$}  & \multicolumn{1}{|@{}c@{}|}{(1 - $\lambda$)$S$ + $\lambda M_1$, (1 - $\gamma$)$T$ + $\gamma S$} \\ 
\cline{2-4}
& \multicolumn{1}{|c}{D} & \multicolumn{1}{|@{}c@{}|}{(1 - $\lambda$)$T$ + $\lambda S$, (1 - $\gamma$)$S$ + $\gamma M_2$} & \multicolumn{1}{|@{}c@{}|}{(1 - $\lambda$)$P$ + $\lambda S$, (1 - $\gamma$)$P$ + $\gamma S$} \\
\cline{2-4}
\end{tabular}\scriptsize\caption{DG - Payoff Matrix Representation}\label{tab:7}
\end{table}

In addition to the inequalities satisfied in the payoffs for PD and SG, we stipulate the two new inequalities below that connect the payoff values from both the PD and the SG:

\begin{center}
$M_1,M_2 > (R + P)/2 \;\;$ and $\;\; T > R > M_ 1\geq  M_2 > P > S$ or $\;\; T > R > M_2 \geq M_1 > P > S$
\end{center}

First, we argue that $M_1$ and $M_2$ should be less than $T$, but greater than $P$. The former should hold, otherwise if $M_1$ and $M_2$ are equal to or greater than $T$, then, there is no dilemma as to what the best strategy is (one should select the highest possible social coefficient and always choose C in order to achieve the highest available payoff), and, the SG loses its meaning. On the other hand, the latter should hold, because, if $M_1$ and $M_2$ are equal to or less than $P$, then, cooperation is discouraged, since one would have no incentive to select a high social coefficient and choose C. In addition, $M_1$ and $M_2$ should be strictly less than $R$, as we would like to encourage cooperation in the SG by assigning to it a payoff value that is somewhat less than the payoff value obtained through mutual cooperation in the PD. This, we believe, reflects more accurately real-life situations, as, in general, the decisions based on moral incentives do not bring high material benefits. Finally, we assume that $M_1$ and $M_2$ should be greater than the average of $R$ and $P$, so that the dilemma of whether to cooperate or defect becomes more intense.

Then, we argue that $M^{'}_{1}$ and $M^{'}_{2}$ should be equal to $S$, so as to discourage defection with a high social coefficient, which would be self-contradictory, and, also, to punish, in a sense, defection, since $M^{'}_{1}$, $M^{'}_{2}$ are the payoff value for defection in the SG, which, by its definition, should not give a high value to defection.

The selection of the social coefficient reveals, in part, the strategy one will follow in a given game. To illustrate this with an example, note that the choice of social coefficient equal to 1 implies cooperation, since defection would give a payoff of 0, and, similarly, the choice of social coefficient equal to 0 most probably implies defection, since cooperation in that case would give a payoff of 0, unless it is mutual, in which case it would be beneficial. On the other hand, selecting a social coefficient between 0 and 1 leaves room for more complex and sophisticated strategies. Finally, as we will see later on, in the implementation of the DG, certain restrictions are imposed on how much a player can increase or decrease the social coefficient in a single round. This is done,  since, in general, humans do not change their moral values radically in a short amount of time.

\subsection{Double game with complete information}
\noindent If we assume that the players know each other's social coefficients prior to every game, the DG becomes a game with complete information and all payoffs are known to both players. In this section we focus on the analysis of the Nash equilibrium in this type of the DG. It is well known that the Nash equilibrium for the PD is mutual defection, represented by (D,D). However, from the perspective of the SG, the best response of any player is to cooperate, as this always leads to a better score as compared to defecting. As a result, if the players make their decisions with no concern for their opponents' behaviour, it leads to a Nash equilibrium of mutual cooperation, represented by (C,C). However, this simplicity cannot be incorporated in the DG, due to the inclusion of the social coefficient, which alters the reward for all outcomes.

In accordance with the equilibrium, we must find out how the social coefficients of the two players alter the potential payoffs for the four possible outcomes of the game. The payoffs for each possible outcome change along with the variation in the social coefficients $\lambda$ and $\gamma$ of players 1 and 2, respectively, as shown in TABLE~\ref{tab:8}. 

At this stage, we consider the payoff equations for player 1. We note that, by symmetry, a similar analysis can be conducted for player 2. Furthermore, without loss of generality, we assume the equality $M^{'}_{1}= M^{'}_{2} = S$, since the SG punishes defection.

For $\lambda$ let us label the three crossing points of the payoff equations as \! $\lambda=a_1$ \! for \! ${\pi}_{1}$(D,D) = ${\pi}_{1}$(C,D), \! $\lambda=b_1$ \! for \! ${\pi}_{1}$(D,C) = ${\pi}_{1}$(C,C) \! and \! $\lambda=c_1$ \! for \! ${\pi}_{1}$(D,C) = ${\pi}_{1}$(C,D). By equating the equations for each payoff, we find the values of the crossing points (similarly for $\gamma$) to be:

\begin{IEEEeqnarray}{rCl}
a_1&=&\frac{P-S}{M_1+P-2S},\ a_2=\frac{P-S}{M_2+P-2S}\\
b_1&=&\frac{T-R}{T-S+M_1-R},\ b_2=\frac{T-R}{T-S+M_2-R}\\
c_1&=&\frac{T-S}{M_1+T-2S},\ c_2=\frac{T-S}{M_2+T-2S}
\end{IEEEeqnarray}
\normalsize

We have the three following cases:

\begin{itemize}
\item  $a_1<b_1<c_1$ ,\ $a_2<b_2<c_2$ \ if \ $P-S<T-R$
\item  $b_1<a_1<c_1$ ,\ $b_2<a_2<c_2$ \ if \ $P-S>T-R$
\item  $a_1=b_1<c_1$ ,\ $a_2=b_2<c_2$ \ if \ $P-S=T-R$
\end{itemize}
\normalsize

We can obtain the Nash equilibrium points for $M_1,M_2 > R$ and $M_1,M_2 < R$ and $M_1= M_2 = R$; note that for all values of $\lambda$ and $\gamma$ that are greater than $b_1$ or $b_2$, the equilibria are equal. To illustrate the method, we will compute below the Nash equilibrium for the two generic cases of $a_1<b_1,\ a_2<b_2$ and $b_1<a_1,\ b_2<a_2$ when $T>R>M_1>M_2>P>M_1^\prime= M_2^\prime=S$.

\subsubsection{Case: $a_1<b_1$ and $a_2<b_2$} 

FIGURE~\ref{fig:1} shows the variation in the payoffs resulting from each outcome of the DG with different values of $\lambda$. We can describe the order of preference for all values of $\lambda$ lying between 0 and 1 by using the values of $a_1, b_1$ and $c_1$ and the functions shown in FIGURE~\ref{fig:1}, and, then, we can obtain the equilibria for different social coefficients. For instance, the preference ordering of player 1 for $0 \leq \lambda < a_1 $ is:
\begin{equation}
  (D,C) > (C,C) > (D,D) > (C,D)
\end{equation}

Since the DG is a symmetric game, similar inequalities also hold for player 2, with the social coefficient being $\gamma$ instead of $\lambda$, and each outcome replaced with its mirror point. For instance, while (D,C) is shown to be the most preferable outcome for player 1 in FIGURE~\ref{fig:1}, (C, D) would take its place for player 2. The pair $(\lambda, \gamma)$ is a point of unit square $[0,1]\times [0,1]$. 

The equilibria for different social coefficients in the case of $a_1 < b_1$ are given in FIGURE~\ref{fig:2} for the generic sub-rectangles and are also presented in TABLE~\ref{tab:8}, which includes the boundary points of these 9 regions. Note that on any boundary point of 2 or 4 generic regions, the set of equilibria is precisely the union of equilibria in the neighbouring generic regions. 

\begin{figure}[h]
\centering
\begin{tikzpicture}
\tikzstyle{every node}=[font=\scriptsize]
  \draw[black,  thick] (0,0) -- (0,5);
    \draw[black,  thick] (0,0) -- (6,0);
      \draw[black,  thick] (6,0) -- (6,5);
        \draw[black,  thick] (0,5) -- (6,5);
       
          \draw[black] (6.6,2.7) -- (8.2,2.7);
    \draw[black] (6.6,2.7) -- (6.6,4.8);
      \draw[black] (6.6,4.8) -- (8.2,4.8);
        \draw[black] (8.2,2.7) -- (8.2,4.8);

 \draw[red,thick, loosely dashed] (6.7,3) -- (7.3,3);
 \draw[black,thick, densely dotted] (6.7,4) -- (7.3,4);
\draw[blue,thick,densely dashed] (6.7,3.5) -- (7.3,3.5);
\draw[yellow,thick] (6.7,4.5) -- (7.3,4.5);

 \draw(7.3,4.5) node[right] {$D,C$};
 \draw(7.3,4) node[right] {$C,C$};
 \draw(7.3,3.5) node[right] {$D,D$};
 \draw(7.3,3) node[right] {$C,D$};

 \draw(0,.5) node[left] {$P$};
 \draw(0,0) node[left] {$S$};
 \draw(0,3.5) node[left] {$M_1$};
 \draw(0,4) node[left] {$R$};
 \draw(0,5) node[left] {$T$};

  \draw(6,.5) node[right] {$P$};
 \draw(6,0) node[right] {$S$};
 \draw(6,3.5) node[right] {$M_1$};
 \draw(6,4) node[right] {$R$};
 \draw(6,5) node[right] {$T$};
 
  \draw(0,0) node[below] {$0$};
   \draw(.75,0) node[below] {$a_1$};
    \draw(1.35,0) node[below] {$b_1$};
     \draw(3.55,0) node[below] {$c_1$};
     \draw(6,0) node[below] {$1$};
 
    \draw(3,5) node[above] {$a_1<b_1<c_1$};
 
 \draw[red,thick, loosely dashed] (0,0) -- (6,3.5);
 \draw[black,thick, densely dotted] (0,4) -- (6,3.5);
\draw[blue,thick,densely dashed] (0,.5) -- (6,0);
\draw[yellow,thick] (0,5) -- (6,0);
 \draw[blue,loosely dotted] (.75,0) -- (.75,.5);
 
  \draw[blue,loosely dotted] (1.35,0) -- (1.35,3.9);
   \draw[blue, loosely dotted] (3.55,0) -- (3.55,2);
\end{tikzpicture}
\caption{Change of payoffs ($T>R>M_1>M_2>P>M_1^\prime= M_2^\prime=S$)}\label{fig:1}
\end{figure}
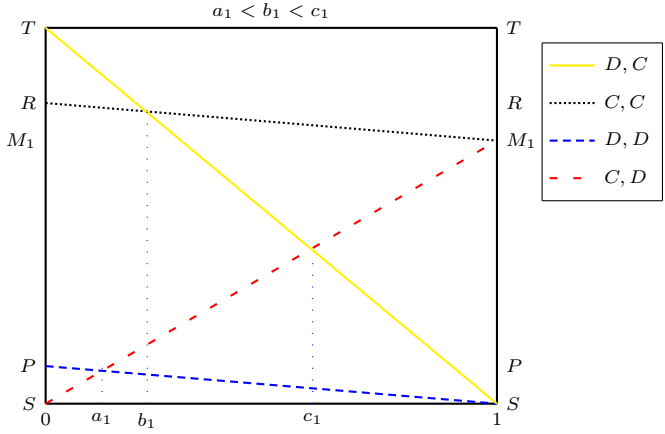


\begin{figure}[h]
\centering
\begin{tikzpicture}
  \draw[black,  thick] (0,0) -- (0,6);
    \draw[black,  thick] (2,0) -- (2,6);
      \draw[black,  thick] (4,0) -- (4,6);
        \draw[black,  thick] (6,0) -- (6,6);
          \draw[black,  thick] (0,0) -- (6,0);
    \draw[black,  thick] (0,2) -- (6,2);
      \draw[black,  thick] (0,4) -- (6,4);
        \draw[black,  thick] (0,6) -- (6,6);

 \draw(0,2) node[left] {$a_2$};
 \draw(0,0) node[left] {$0$};
 \draw(0,4) node[left] {$b_2$};
 \draw(0,6) node[left] {$1$};
 
  \draw(0,0) node[below] {$0$};
   \draw(2,0) node[below] {$a_1$};
    \draw(4,0) node[below] {$b_1$};
     \draw(6,0) node[below] {$1$};
 
\node at (1,5) {(D,C)};
\node at (1,3) {(D,C)};
\node at (1,1) {(D,D)};

\node at (3,5) {(D,C)};
\node at (3,3) {(C,D),(D,C)};
\node at (3,1) {(C,D)};

\node at (5,5) {(C,C)};
\node at (5,3) {(C,D)};
\node at (5,1) {(C,D)};
\end{tikzpicture}
\caption{The set of equilibria for each of the 9 generic regions ($a_1<b_1,\ a_2<b_2$)}\label{fig:2}
\end{figure}

\begin{table}[h]\tiny\renewcommand{\arraystretch}{2}

\centering
\centering\begin{tabular} {|@{}c@{}|@{}c@{}|@{}c@{}|@{}c@{}|@{}c@{}|@{}c@{}|} \hline 

$b_2<\gamma\leq1$ & (D,C) & (D,C) & (D,C) & (C,C),(D,C) & (C,C)   \\ \hline 
$\gamma=b_2$ & (D,C) & (C,D),(D,C)  & (C,D),(D,C)& (C,C),(C,D),(D,C)& (C,C),(C,D)\\ \hline 
$a_2<\gamma<b_2$ & (D,C) & (C,D),(D,C) & (C,D),(D,C) & (C,D),(D,C)&(C,D) \\ \hline 
$\gamma=a_2$ & (D,D),(D,C) & (D,D),(D,C),(C,D) & (C,D),(D,C)& (C,D),(D,C)& (C,D) \\ \hline 
$0\leq\gamma<a_2$ & (D,D) & (D,D),(C,D) & (C,D) & (C,D)&(C,D) \\ \hline 
~ &$0\leq\lambda<a_1$ & $\lambda=a_1$ & $a_1<\lambda<b_1$ & $\lambda=b_1$ &$b_1<\lambda\leq1$ \\ \hline 
\end{tabular}\caption{The equilibria for different social coefficients for $a_1<b_1$ and $a_2<b_2$ }\label{tab:8}
\end{table}

\subsubsection{Case: $b_1<a_1$ and $b_2<a_2$}

FIGURE~\ref{fig:3} illustrates the change in payoffs resulting from each outcome of the DG with different values of $\lambda $ in the case of $b_1<a_1$. FIGURE~\ref{fig:4} provides the set of equilibria for each of the 9 generic regions and TABLE~\ref{tab:9} presents the equilibria for all possible social coefficients.


\begin{figure}[h]
\centering
\begin{tikzpicture}
\tikzstyle{every node}=[font=\scriptsize]
  \draw[black,  thick] (0,0) -- (0,5);
    \draw[black,  thick] (0,0) -- (6,0);
      \draw[black,  thick] (6,0) -- (6,5);
        \draw[black,  thick] (0,5) -- (6,5);
       
          \draw[black] (6.6,2.7) -- (8.2,2.7);
    \draw[black] (6.6,2.7) -- (6.6,4.8);
      \draw[black] (6.6,4.8) -- (8.2,4.8);
        \draw[black] (8.2,2.7) -- (8.2,4.8);

 \draw[red,thick, loosely dashed] (6.7,3) -- (7.3,3);
 \draw[black,thick, densely dotted] (6.7,4) -- (7.3,4);
\draw[blue,thick,densely dashed] (6.7,3.5) -- (7.3,3.5);
\draw[yellow,thick] (6.7,4.5) -- (7.3,4.5);

 \draw(7.3,4.5) node[right] {$D,C$};
 \draw(7.3,4) node[right] {$C,C$};
 \draw(7.3,3.5) node[right] {$D,D$};
 \draw(7.3,3) node[right] {$C,D$};

 \draw(0,3) node[left] {$P$};
 \draw(0,0) node[left] {$S$};
 \draw(0,3.5) node[left] {$M_1$};
 \draw(0,4) node[left] {$R$};
 \draw(0,5) node[left] {$T$};

  \draw(6,3) node[right] {$P$};
 \draw(6,0) node[right] {$S$};
 \draw(6,3.5) node[right] {$M_1$};
 \draw(6,4) node[right] {$R$};
 \draw(6,5) node[right] {$T$};
 
  \draw(0,0) node[below] {$0$};
   \draw(2.78,0) node[below] {$a_1$};
    \draw(1.35,0) node[below] {$b_1$};
     \draw(3.55,0) node[below] {$c_1$};
     \draw(6,0) node[below] {$1$};
 
    \draw(3,5) node[above] {$b_1<a_1<c_1$};
 
 \draw[red,thick, loosely dashed] (0,0) -- (6,3.5);
 \draw[black,thick, densely dotted] (0,4) -- (6,3.5);
\draw[blue,thick,densely dashed] (0,3) -- (6,0);
\draw[yellow,thick] (0,5) -- (6,0);
 \draw[blue,loosely dotted] (2.78,0) -- (2.78,1.6);
 
  \draw[blue,loosely dotted] (1.35,0) -- (1.35,3.9);
   \draw[blue,loosely dotted] (3.55,0) -- (3.55,2);
\end{tikzpicture}
\caption{Change of payoffs (The variation of payoffs ($T>R>M_1>M_2>P>M_1^\prime= M_2^\prime=S$)}\label{fig:3}
\end{figure}
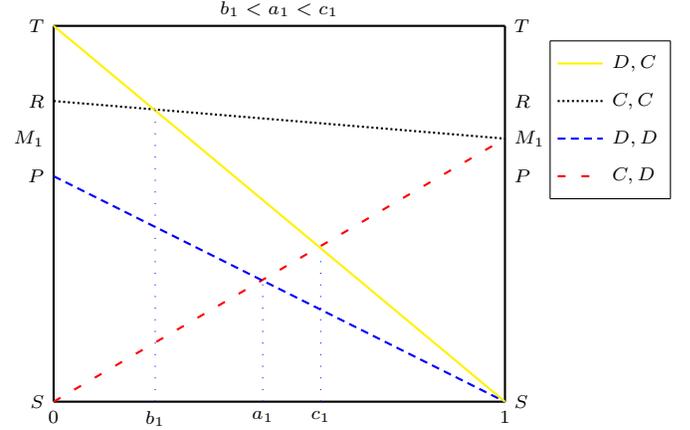


\begin{figure}[h]
\centering
\begin{tikzpicture}
  \draw[black,  thick] (0,0) -- (0,6);
    \draw[black,  thick] (2,0) -- (2,6);
      \draw[black,  thick] (4,0) -- (4,6);
        \draw[black,  thick] (6,0) -- (6,6);
          \draw[black,  thick] (0,0) -- (6,0);
    \draw[black,  thick] (0,2) -- (6,2);
      \draw[black,  thick] (0,4) -- (6,4);
        \draw[black,  thick] (0,6) -- (6,6);

 \draw(0,2) node[left] {$a_2$};
 \draw(0,0) node[left] {$0$};
 \draw(0,4) node[left] {$b_2$};
 \draw(0,6) node[left] {$1$};
 
  \draw(0,0) node[below] {$0$};
   \draw(2,0) node[below] {$a_1$};
    \draw(4,0) node[below] {$b_1$};
     \draw(6,0) node[below] {$1$};
 
\node at (1,5) {(D,C)};
\node at (1,3) {(D,D)};
\node at (1,1) {(D,D)};

\node at (3,5) {(C,C)};
\node at (3,3) {(C,C),(D,D)};
\node at (3,1) {(D,D)};

\node at (5,5) {(C,C)};
\node at (5,3) {(C,C)};
\node at (5,1) {(C,D)};
\end{tikzpicture}
\caption{The set of equilibria for each of the 9 generic regions ($b_1<a_1,\ b_2<a_2$)}\label{fig:4}
\end{figure}
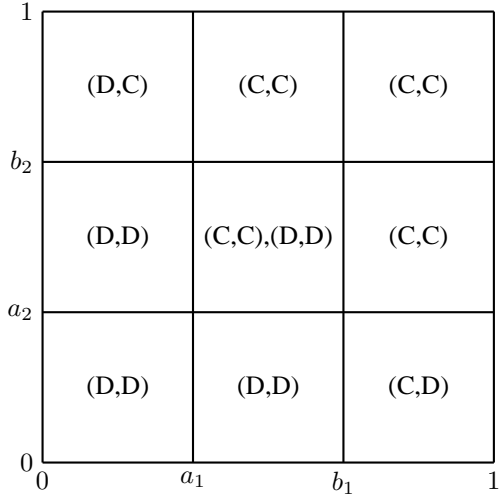

\begin{table}[h]\tiny\renewcommand{\arraystretch}{2}

\centering\begin{tabular} {|@{}c@{}|@{}c@{}|@{}c@{}|@{}c@{}|@{}c@{}|@{}c@{}|} \hline 
$a_2<\gamma\leq1$ & (D,C) & (D,C) & (C,C) & (C,C) & (C,C)   \\ \hline 
$\gamma=a_2$ & (D,C) & (D,C),(D,D),(C,C)  & (D,D),(C,C) & (D,D),(C,C)& (C,C)\\ \hline 
$b_2<\gamma<a_2$ & (D,D) & (D,D),(C,C) & (D,D),(C,C) & (D,D),(C,C)&(C,C)\\ \hline 
$\gamma=b_2$ & (D,D) & (D,D),(C,C) & (D,D),(C,C) & (C,D),(D,D),(C,C)& (C,D) \\ \hline 
$0\leq\gamma<b_2$ & (D,D) & (D,D) & (D,D) & (C,D) & (C,D) \\ \hline 
~ & $0\leq\lambda<b_1$&$\lambda=b_1$&$b_1<\lambda<a_1$&$\lambda=a_1$&$a_1<\lambda\leq1$\\ \hline
\end{tabular}\caption{ The equilibria for different social coefficients for $b_1<a_1$ and $b_2<a_2$}\label{tab:9}
\end{table}

\subsection{Double game with incomplete information}

\noindent  We now assume that the players do not know each other's social coefficients prior to any game, which means that they do not know the full values of the payoff matrix. Thus, the game has two-sided \textit{incomplete information}. We assume that the social coefficient of each player has a finite number of possible values and that the probability distribution of the social coefficient is common knowledge between the two players. We also assume that, at the start of the game, each player is aware of the value of their own social coefficient (\textit{private information}) or {\em type}, but not the value of the social coefficient of the opponent. In that way, each player relies on a probabilistic inference to predict the opponent's actions. From TABLES~\ref{tab:8} and~\ref{tab:9}, we see immediately that for extreme types $\lambda,\gamma=0,1$, we have the pure NE $(D,D)$, $(C,D)$, $(D,C)$ and $(C,C)$ and it follows immediately that in both cases of $a_1<b_1$, $a_2<b_2$, and  $b_1<a_1$, $b_2<a_2$, we have pure regular games.

We present two specific examples with finite sets of types for the two players, one of which is actually the basis for the round-robin tournament in the next section as well. Assume\begin{center}
$T>R>M>P>M^\prime=S$ \, and \, $a_1<b_1$, $a_2<b_2$\end{center}

\subsubsection{Example I}

We choose the four discrete values, or types,  

\begin{IEEEeqnarray}{rCl}
\lambda_1&=&0, \,\lambda_2=a_1, \,\lambda_3= b_1, \,\lambda_4= 1\nonumber\\
\gamma_1&=&0, \,\gamma_2=a_2, \,\gamma_3= b_2, \,\gamma_4= 1\nonumber
\end{IEEEeqnarray}

Note that $\lambda_2,\gamma_2,\lambda_3,\gamma_3$ give the values at the boundaries of the three generic regions in the unit square with each other. Table~\ref{tab:10} gives the set of Nash equilibria for all possible pairs $(\lambda_m,\gamma_n)$.

\begin{table}[ht]\footnotesize\renewcommand{\arraystretch}{2}
\centering\tiny\begin{tabular}{|c|c|c|c|c|}\hline
$\gamma_4$ &  (D,C) & (D,C) & (C,C),(D,C) & (C,C)   \\ \cline{1-5}
$\gamma_3$ & (D,C) & (C,D),(D,C)  & (C,D),(D,C),(C,C)& (C,C),(C,D) \\ \cline{1-5}
$\gamma_2$ & (D,D),(D,C) & (D,D),(D,C),(C,D) & (C,D),(D,C)& (C,D) \\\cline{1-5} 
$\gamma_1$ & (D,D) & (D,D),(C,D) & (C,D)&(C,D) \\ \hline 
  & $\lambda_1$ & $\lambda_2$ & $\lambda_3$ & $\lambda_4$ \\ \hline  
\end{tabular}\caption{ The equilibria for different social coefficients for $a_1<b_1$ and $a_2<b_2$ with four types per player}\label{tab:10}
\end{table}

From this table, we see that for pairs of types where there is a choice of pure NE, we can choose a pure NE such that we obtain TABLE~\ref{tab:11}

\begin{table}[ht]\footnotesize\renewcommand{\arraystretch}{2}
\renewcommand{\arraystretch}{2}
\tiny\begin{center}
\begin{tabular}{|c|c|c|c|c|}\hline

$\gamma_4$ &  (D,C) & (D,C) & (C,C) & (C,C)   \\ \cline{1-5}
$\gamma_3$ & (D,C) & (D,C)  & (C,C)& (C,C) \\ \cline{1-5}
$\gamma_2$ & (D,D) & (D,D)& (C,D)& (C,D) \\\cline{1-5} 
$\gamma_1$ & (D,D) & (D,D)& (C,D)&(C,D) \\ \hline 
  & $\lambda_1$ & $\lambda_2$ & $\lambda_3$ & $\lambda_4$ \\ \hline  
\end{tabular}\caption{{ The equilibria  chosen for different social coefficients from TABLE~\ref{tab:10}}}\label{tab:11}
\end{center}
\end{table}

From the above, we see that the DG is completely pure regular with $(DDCC,DDCC)$ as a pure Bayesian NE.

\subsubsection{Example II}

We take 5 discrete values, or types for each player as follows,  
\begin{IEEEeqnarray}{rCl}
\lambda_1&=&0, \,\lambda_2=a_1, \,\lambda_3= \frac{a_1 + b_1}{2}, \,\lambda_4= b_1, \,\lambda_5= 1\nonumber\\
\gamma_1&=&0, \,\gamma_2=a_2, \,\gamma_3= \frac{a_2 + b_2}{2}, \,\gamma_4= b_2, \,\gamma_5= 1\nonumber
\end{IEEEeqnarray}

Table~\ref{tab:12} gives the set of Nash equilibria for all possible pairs $(\lambda_i,\gamma_i)$. 
\begin{table}[ht]\footnotesize\renewcommand{\arraystretch}{2}
\renewcommand{\arraystretch}{2}
\tiny\begin{center}
\begin{tabular}{|c|c|c|c|c|c|}\hline

$\gamma_5$ &  (D,C) & (D,C) & (D,C) & (C,C),(D,C) & (C,C)   \\ \cline{1-6}
$\gamma_4$ & (D,C) & (C,D),(D,C)  & (C,D),(D,C)& (C,D),(D,C),(C,C)& (C,C),(C,D) \\ \cline{1-6}
$\gamma_3$ & (D,C) & (C,D),(D,C) & (C,D),(D,C) & (C,D),(D,C)&(C,D) \\ \cline{1-6}
$\gamma_2$ & (D,D),(D,C) & (D,D),(D,C),(C,D) & (C,D),(D,C)& (C,D),(D,C)& (C,D) \\\cline{1-6} 
$\gamma_1$ & (D,D) & (D,D),(C,D) & (C,D) & (C,D)&(C,D) \\ \hline 
  & $\lambda_1$ & $\lambda_2$ & $\lambda_3$ & $\lambda_4$ & $\lambda_5$ \\ \hline  
\end{tabular}\caption{{ The equilibria for different social coefficients for $a_1<b_1$, $a_2<b_2$ with five types per player. }}\label{tab:12}
\end{center}
\end{table}

From this table, we can see that the DG is not completely pure regular. We will use this DG in the next section for the repeated PD game.

\section{Round Robin Tournament Analysis}
\noindent A computer tournament of the DG has been implemented to operate as a framework,  where we could test the validity of the theoretical results, by comparing the performance of various competing strategies. The structure of the tournament is round-robin, and, thus, all the strategies compete in the iterated double game against all the other strategies and themselves once. Each game between any two strategies consists of 200 rounds and the total score of a strategy in a game is the sum of the payoffs acquired from all the rounds. The average score of a strategy accounts for its robustness and stability, since the amount of strategies ensures the existence of a rather competitive environment. The numerical values we use for the payoff values are the following: 
\begin{displaymath}
T = 5, \,\, R = 3, \,\,  M_1 = M_2 = 2.5, \,\,  P = 1, \,\,  M'_1 = M'_2 = S = 0 \,\, 
\end{displaymath}
These values are chosen precisely to create an environment, where a real dilemma is present. Given the aforementioned values for the payoffs, a strategy can score between 0 and 1000 in a single game. 

A score of 0 can be obtained by either a strategy that has a social coefficient equal to 0 and cooperates throughout the whole game, while the opponent only defects, or a strategy that has a social coefficient equal to 1 and defects throughout the whole game, irrespective of what its opponent does. 

On the other hand, a score of 1000 can only be obtained by a strategy that has a social coefficient equal to 0 and defects throughout the whole game, while the opponent only cooperates. 

A score of 200 can be obtained by two strategies that have social coefficients equal to 0 and mutually defect throughout the whole game. 

A score of 500 can be obtained by a strategy that has a social coefficient equal to 1 and cooperates throughout the whole game, irrespective of what the opponent does. 

Finally, a score of 600 can be obtained by two strategies that have social coefficients equal to 0 and mutually cooperate throughout the whole game. 

The scores for a single game described here give us an idea of what to expect from the games, and, as we will see, the winner of the tournament has an average score of higher than 600.

A significant aspect of the tournament is the selection of a social coefficient by a strategy. As mentioned, we allow the social coefficient to be part of the discrete set of five distinct values, that we used in the theoretical analysis as in example II,

\begin{displaymath}
0, \,\, a, \,\, \frac{a + b}{2}, \,\, b, \,\, 1
\end{displaymath}

The strategies are allowed to change their social coefficients within the game and adapt them to the environment they face. This is done due to the fact that the repeated DG for the PD is assumed to reflect a range of situations, spanned over a finite interval of time, represented by the number of rounds of the game. As a result, since human beings can change their social values, according to the experiences that they acquire, it seems realistic to have a varying social coefficient for the repeated version of the DG for the PD. 

However, since it is difficult to quantify exactly by how much human beings change their social values, we only allow the strategies to change their social coefficients stepwise, and, as a result, they can either increase them or decrease them by one value at any round. This is done to avoid having strategies changing their social coefficients from a value of 0 to that of 1 in a single round, since, it is believed that, only under extreme and unprecedented circumstances would such a sudden change occur in one's social values.

The strategies participating in the tournament vary in ways such as their algorithmic complexity, the choice they make for the first round of the game and their initial social coefficient. Some of them take into account the decisions that their opponent has made up to the point of consideration in the game, some use probabilistic estimations and even randomness in making their decisions, some have already made up their mind and follow rules that do not change according to the flow of the game, some make use of Bayesian inference and others use the $(\lambda, \gamma)$ diagrams to make their decisions according to the Nash equilibria indications. In essence, a strategy consists of an algorithm and so it operates according to certain instructions, changes the social coefficient and provides the decision of whether to cooperate or defect.

The initial social coefficient of a strategy shows its intentions, since a low social coefficient usually implies proneness to defection, while a high social coefficient on the other hand cooperation. Varying initial social coefficients across tournaments mean changing initial conditions, and, as a result, dynamic environments. Certain strategies have complex ways for dealing with their opponents' initial behaviour, and, so, what they may infer from it, may, in some cases, pre-determine the rest of the course of the game. In addition, most strategies have algorithms that modify their social coefficient in almost every round, thus, enabling them to achieve in the course of the game their optimal social coefficient for a particular game and adapt to the environment that has been developed from their opponents' actions. Then, they can respond effectively to both cooperative and defective behaviours and not be restricted by their choice of initial social coefficient.

As mentioned in the theoretical part of the analysis of the DG, different kinds of behaviour are observed when $\lambda$ and $\gamma$ change. 

With $0 \le \lambda \le a$ and $0 \le \gamma \le a$, the Nash equilibrium is provided by (D, D), i.e. mutual defection. So, for such social coefficients, we expect to see defective behaviours as strategies try to recognise their opponents' intentions and see whether they can get away with defection, or they will face retaliatory behaviour. 

With $a < \lambda < b$ and $a < \gamma < b$, the Nash equilibria are provided by (C, D) and (D, C), i.e. a game of brinkmanship. So, for such social coefficients, the player who defects first has an advantage and dominates by gaining from an opponent's cooperation. However, with the increase of the social coefficient, cooperating can be beneficial, since a player can gain the reward from the SG. In that way, a lot of strategies change their behaviour at this stage and employ a cooperative approach to the game. 

With $b \le \lambda \le 1$ and $b \le \gamma \le 1$, the Nash equilibrium is provided by (C, C), i.e. mutual cooperation. So, for such social coefficients, strategies with cooperative behaviour can gain the social rewards and not suffer the social punishments from the SG. The various stages of the social coefficients' changes that alter the form of the DG are summarised below.\\

\begin{itemize}
\item $0 \le \lambda \le a,\,\, 0 \le \gamma \le a$:\\ Nash equilibria: (D, D)\\  PD
\item $a < \lambda < b, \,\, a < \gamma < b$:\\ Nash equilibria: (C, D), (D, C)\\ Chicken Game / Game of Brinkmanship
\item $b \le \lambda \le 1, \,\, b \le \gamma \le 1$:\\ Nash equilibria (C, C)\\ Cooperation Game
\end{itemize}

\section{SEG}

\noindent The results of the tournament showed a clear winning strategy, whose average and cumulative scores were much higher than those of any other participating strategy. Its algorithm is a mixture of the results of the theoretical work and some conditions on how to alter its social coefficient, so as to adapt to the course of action of any game. This strategy was devised to take the maximum benefits of a game, irrespective of the nature of the opposing strategies. 

We call it SEG and explain its functionality as follows. SEG is based on two parts; deciding whether to cooperate or defect and altering its social coefficient based on some pre-defined conditions. For the former, it looks up the $(\lambda, \gamma)$ diagram and behaves as the Nash equilibria indicate, thus, its decision of whether to cooperate or defect depends only on the theoretical work and the results drawn from it. For the latter, it changes its social coefficient according to the following conditions:
\begin{itemize}
\item If SEG chose C and its opponent chose C in the previous round, it does not change its social coefficient.
\item If SEG chose C and its opponent chose D in the previous round, it increases its social coefficient.
\item If SEG chose D and its opponent chose C in the previous round, it decreases its social coefficient.
\item If SEG chose D and its opponent chose D in the previous round, it increases its social coefficient.
\end{itemize}

It should be also mentioned that its initial social coefficient is 0, the result of which is that in the first round it defects, since this is indicated by the Nash equilibria in the $(\lambda, \gamma)$ diagram. 

Its success can be attributed to several reasons. It works on the principle of adjusting its social coefficient based not only on its opponent's behaviour, but its own as well. If its opponent defected and it either cooperated or defected in the previous round, it increases its social coefficient to avoid the disastrous cycles of mutual defection, that would be caused by a low social coefficient, since that would be the indication from the Nash equilibria in the $(\lambda, \gamma)$ diagram.
In that way, SEG, instead of trying to alter its opponent's behaviour by retaliating to its defections, realises that it is better off by increasing its social coefficient to get the social rewards from cooperation. If its opponent cooperated and it defected in the previous round, it decreases its social coefficient, as it sees that the opponent does not retaliate to its defections, so it exploits the situation and behaves as instructed from the Nash equilibria in the $(\lambda, \gamma)$ diagram, so as to get the temptation's payoff to defect. In a sense it exploits the opponent in that particular case, however, it is the opponent who allows that behaviour by not retaliating. In such cases, some strategies allow themselves to be taken advantage of, because, if they have a high social coefficient, they can get the social rewards from cooperation and not concern themselves about their opponents' actions. Finally, if both its opponent and it cooperated in the previous round, it keeps its social coefficient, so as to make this cooperation stable and sustainable.

Although SEG seems a strategy that promotes defective behaviour, that is not the case, as with all the strategies that are cooperative and do not let themselves to be taken advantage of, cooperation emerges quickly. Also, in its choice of whether to cooperate or defect, it behaves as specified by the Nash equilibria in the $(\lambda, \gamma)$ diagram, so it does not behave oddly and unrealistically, such as defecting with a high social coefficient or cooperating with a low social coefficient. It should be also noted that in the points, where there exist multiple Nash equilibria, and there is a choice between cooperating and defecting, we follow our theoretical analysis and select defection, as this potentially maximises the payoff, given the particular payoff values.

Another fact that shows the success of SEG is that we conducted two separate tournaments, where we changed the values of M, the social rewards payoff for cooperation, to 3 and 2, respectively, to see whether the value of 2.5, which was our initial selection, biased the results. The result was that, in both cases, SEG was the winning strategy of the tournament.

The rankings, the total scores of the 10 top-scoring strategies,  the initial social coefficients and their decisions of whether to cooperate or to defect in the first round (when the choice is specified by the Nash equilibria of the $(\lambda, \gamma)$ diagram, we write NE) are presented in TABLE~\ref{tab:T10SSTSISCAIC}.\footnote{\label{fn:1}Please note that the algorithms of the strategies participating in the round-robin tournament are not presented here due to space constraints.}

\begin{table}[h!]\renewcommand{\arraystretch}{1.3}\scriptsize
\centering
\begin{tabular}{c c c c c} \hline
Ranking & Strategy & Total Score & Initial Social Coeff. & Initial Choice\\ \hline
1 & SEG & 121536.49 & 0.00 & NE\\
2 & ANE & 114694.48 & 0.00 & NE\\
3 & BNE & 113751.58 & 0.00 & NE\\
4 & ADAPTIVE & 112387.83 & 0.00 & C/D\\
5 & TQC & 109570.71 & 0.29 & D\\
6 & MTFT & 108649.00 & 0.00 & D\\ 
7 & MIXED & 103372.31 & 1.00 & C\\
8 & SD & 97643.53 & 0.44 & C/D\\
9 & DTFT & 97330.00 & 0.00 & C\\ 
10 & TFT & 96805.00 & 0.00 & C\\ [1ex] \hline
\end{tabular}
\caption{Top 10 Scoring Strategies, Total Scores, Initial Social Coefficients and Initial Choices}
\label{tab:T10SSTSISCAIC}
\end{table}
\normalsize
The rankings and the average scores of the 10 top-scoring strategies are presented in table~\ref{tab:T10SSAS}.\footnote{Same as 1.}

\begin{table}[h!]\renewcommand{\arraystretch}{1.3}
\centering
\begin{tabular}{c c c} \hline
Ranking & Strategy & Average Score\\ \hline
1 & SEG & 633.00\\
2 & ANE & 597.37\\
3 & BNE & 592.46\\
4 & ADAPTIVE & 585.35\\
5 & TQC & 570.68\\
6 & MTFT & 565.88\\
7 & MIXED & 538.40\\
8 & SD & 508.56\\
9 & DTFT & 506.93\\ 
10 & TFT & 504.19\\ [1ex] \hline
\end{tabular}
\caption{Top 10 Scoring Strategies, Average Scores}
\label{tab:T10SSAS}
\end{table}

Let us now isolate some of the games of SEG and examine them, so that its strengths in getting the maximum from a game, will become apparent.\\
\begin{itemize}
\item SEG vs. ALLD:\\

ALLD is a strategy that has an initial social coefficient equal to 0, it constantly chooses D and never changes its social coefficient.\\
The game decisions between the two are as follows:\\

SEG: \;\, D \, D \, D \, C \, D \, C \, C \, C \, \! C \, C \dots \! C \dots \! C\\
ALLD: D \, D \, D \, D \, D \, D \, D \, D \, D \, D \dots \! D \dots \! D\\

As we can see, since SEG initially has a social coefficient equal to 0, it chooses what the Nash equilibria of the $(\lambda, \gamma)$ diagram indicate, which, for such a social coefficient, is D. However, although other strategies would keep choosing D against ALLD, and thus, get hurt from mutual defection, the algorithm of SEG makes it increase its social coefficient, and, from the seventh round and onwards, it has a very high social coefficient, and, thus, chooses C for the rest of the game. It gets the maximum possible of points from a clearly uncooperative strategy like ALLD.\\

\item SEG vs. ALLC:\\

ALLC is a strategy that has an initial social coefficient equal to 1, it constantly chooses C and never changes its social coefficient.\\
The game decisions between the two are as follows:\\

SEG: \; D \, D \, D \, D \, D \, D \, D \, D \, D \, D \dots \! D \dots \! D\\
ALLC: C \, C \, C \, C \, C \, C \, \! C \, \! C \, C \, \! C \dots \! C \dots \! C\\

As we can see, since SEG initially has a social coefficient equal to 0, it chooses what the Nash equilibria of the $(\lambda, \gamma)$ diagram indicate, which, for such a social coefficient, is D. Then, it realizes that ALLC will not retaliate to its defections and keeps choosing D, so that it can constantly acquire the temptation to defect payoff. In some sense, it exploits the unwillingness of ALLC to retaliate, however, ALLC gets the social reward for cooperation, so it is not entirely taken advantage of, since a high social coefficient allows a strategy to be indifferent towards the opponent's actions. Again, SEG gets the maximum possible of points from a cooperative, that does not retaliate, strategy like ALLC.\\

\item SEG vs. TFT:\\

TFT is a strategy that has an initial social coefficient equal to 1, it initially chooses C, and for the rest of the game chooses what its opponent chose in the last round. It does not change its social coefficient.\\
The game decisions between the two are as follows:\\

SEG: \! D \, D \, D \, D \, C \, C \, C \, C \, C \, C \dots \! C \dots \! C\\
TFT: \, C \, D \, D \, D \, D \, C \, C \, C \, C \, C \dots \! C \dots \! C\\

As we can see, since SEG initially has a social coefficient equal to 0, it chooses what the Nash equilibria of the $(\lambda, \gamma)$ diagram indicate, which, for such a social coefficient, is D. TFT initially chooses C, but then retaliates to the defection of SEG and chooses D. Thus, although after the first round, SEG decreased its social coefficient, after the second round and the retaliatory behaviour of TFT, it increases its social coefficient, so as to avoid the destructive consequences of mutual defection. Then, after the fifth round, it has increased its social coefficient enough, so that the Nash equilibria indicate C as the choice for the next round. TFT reciprocates both cooperation and defection, and, thus, chooses C in the next round, as well. Therefore, for the remaining rounds, both strategies choose C, and, cooperation based on reciprocity has emerged. Again, SEG gets the maximum possible points from an opponent that is cooperative, but protects itself, and retaliates to defection, as well.
\end{itemize}

The conclusion to be drawn is that SEG seems to be a strategy that performs well under numerous environments and it shows signs of robustness and stability. It successfully puts the theoretical results into practice, since it uses the Nash equilibria, single or multiple, of the $(\lambda, \gamma)$ diagrams to make its decisions. It uses its social coefficient to promote cooperation based on reciprocity, to protect itself from purely defective strategies, and to exploit situations, where it realises that there is no willingness for retaliation from the opponent. Finally, although it starts by choosing D and by having a low initial social coefficient, it can quickly change its behaviour, should it see signs for cooperation based on reciprocity from the opponent's side.

\section{Conclusion}

 In this paper, we have introduced Multi-Games (MG) which can be used when a number of players need to divide up their resources according to different weights to a given number of games which are then played out simultaneously. For the case of a two-player Double Game (DG) in which each player has the same set of strategies for the two basic games, we have shown that when the possible weights for each player belong to a finite set of types, it is possible to decide in linear time if there is a set of local pure NE for all pairs of types, for which the strategy of each player only depends on its own type. In such a case we immediately obtain a Bayesian pure NE for the DG. 

This result, in future work, will be extended to mixed NE on the one hand and to any N-player MG on the other. A challenging question is if we can reduce the complexity of computing a Bayesian NE for a pure regular (but not completely pure regular) MG. 

We also extended the classical PD with a Social Game (SG) to obtain a DG which can model the prosocial behaviour of human agents. We showed that two types of DG are possible with some reasonable assumptions about the SG. For one of these two possibilities, we provided two examples for DG with finite sets of types for the players, of which the first was completely pure regular, with a pure Bayesian NE which was directly determined, while the second was not.  
  
Finally, we conducted a tournament consisting of various strategies for the repeated DG of the second example above. This provided a testing framework for our work, and we devised a winning strategy. In this repeated version of the DG for future work, we could keep, in each round of a game, separate scores  for the material and social components of each strategy. Then at the end of the final game, we can take a convex combination of the two total scores by applying a global weight, which would be determined by society as a whole. This we believe would add to the degree of realism that the DG aims to achieve.

\bibliography{bibliography}

\end{document}